\def\year{2016}\relax
\documentclass[letterpaper]{article}
\usepackage{aaai16}
\usepackage{times}
\usepackage{helvet}
\usepackage{courier}

\usepackage{amsfonts,amsmath,amsthm}
\usepackage{algorithm}
\usepackage[noend]{algorithmic}
\usepackage{graphicx}

\usepackage{color}

\begin{document}
%
\title{Approximation Algorithms for Route Planning with Nonlinear Objectives}
\author{Ger Yang\\
Electrical and Computer Engineering\\
University of Texas at Austin\\
geryang@utexas.edu\\
\And
Evdokia Nikolova\\
Electrical and Computer Engineering\\
University of Texas at Austin\\
nikolova@austin.utexas.edu
}


\newtheorem{theorem}{Theorem}
\newtheorem{definition}{Definition}
\newtheorem{lemma}{Lemma}
\newtheorem{proposition}{Proposition}
\newtheorem{conjecture}{Conjecture}
\newtheorem{question}{Question}
\newtheorem{corollary}{Corollary}
\newtheorem{property}{Property}
\newtheorem{example}{Example}
\newtheorem{remark}{Remark}
\newtheorem{assumption}{Assumption}
\newtheorem{fact}{Fact}

\newcommand{\paths}{\mathcal{P}}
\newcommand{\R}{\mathbb{R}}
\newcommand{\N}{\mathbb{N}}
\newcommand{\Wmat}{\mathbf{W}}
\newcommand{\wvec}{\mathbf{w}}
\newcommand{\yvec}{\mathbf{y}}
\newcommand{\tfor}{\text{ for }}
\newcommand{\Wspace}{\mathcal{W}}
\newcommand{\Cregion}{\mathcal{C}}
\newcommand{\avec}{\mathbf{a}}
\newcommand{\Gset}{\mathcal{G}}
\newcommand{\lvec}{\mathbf{l}}
\newcommand{\hvec}{\mathbf{h}}

\newcommand{\upd}[1]{\textcolor[rgb]{1,0,0}{#1}}

\nocopyright

\maketitle
\begin{abstract}
We consider optimal route planning when the objective function is a general nonlinear and non-monotonic function.  Such an objective models user behavior more accurately, for example, when a user is risk-averse, or the utility function needs to capture a penalty for early arrival.  It is known that as nonlinearity arises, the problem becomes NP-hard and little is known about computing optimal solutions when in addition there is no monotonicity guarantee.  We show that an approximately optimal non-simple path can be efficiently computed  under some natural constraints.  
In particular, we provide a fully polynomial approximation scheme under hop constraints.  Our approximation algorithm can extend to run in pseudo-polynomial time under a more general linear constraint that sometimes is useful.  As a by-product, we show that our algorithm can be applied to the problem of finding a path that is most likely to be on time for a given deadline.
\end{abstract}

\section{Introduction}
In this paper, we present approximation algorithms for route planning with a nonlinear objective function.  Traditionally, route planning problems are modeled as linear shortest path problems and there are numerous algorithms, such as Bellman-Ford or Dijkstra's algorithm, that are able to find the optimal paths efficiently.  However, nonlinearity arises when we would like to make more complicated decision-making rules.  For example, in road networks, the objective might be a nonlinear function that represents the trade-off between traveling time and cost.  Unfortunately, dynamic programming techniques used in linear shortest path algorithms no longer apply as when nonlinearity appears, 
subpaths of optimal paths may no longer be optimal.

Applications include risk-averse routing or routing with a nonlinear objective when we are uncertain about the traffic conditions \cite{Nikolova:2006ab,Nikolova:2006aa}.  A risk-averse user tends to choose a route based on both the speediness and reliability.  Such risk-averse attitude can often be captured by a nonlinear objective, e.g. a mean-risk objective \cite{Nikolova:2010aa}.  On the other hand, the user might seek a path that maximizes the probability for him to be on time given a deadline. 
This objective was considered by \citeauthor{Nikolova:2010aa}~\shortcite{Nikolova:2010aa} to capture the 
risk-averse behavior.  For Gaussian travel times, the latter paper only considered the case where at least one path has mean travel time less than the deadline, and consequently the objective function can be modeled as a monotonic function.

The problem of finding a path that minimizes a monotonic increasing objective has been studied in the literature \cite{Goyal:2013aa,Nikolova:2010aa,Tsaggouris:2009aa}.  What if there is no monotonicity guarantee?  For example, consider the deadline problem when we know that all paths have mean travel time longer than the deadline, or consider the well-known cost-to-time ratio objective \cite{Megiddo:1979aa,Guerin:1999aa} in various combinatorial optimization problems.  Without the monotonicity assumption, the problem usually becomes very hard.  It has been shown \cite{Nikolova:2006ab} that for a general objective function, finding the optimal {\em simple} path is NP-hard.  This is because such a problem often involves finding the longest path for general graphs, which is known to be a strongly NP-hard problem \cite{Karger:1997aa}, namely finding a constant-factor approximation is also NP-hard.  This suggests that for a very general class of functions, including the cost-to-time ratio objective \cite{Ahuja:1983aa}, not only is the problem itself NP-hard but so is also its approximation counterpart.  

Therefore, in this paper, we focus on the problem that accepts \emph{non-simple} paths, in which it is allowed to visit nodes more than once. 
To make this problem well-defined, we consider the problem under either a hop constraint or an additional linear constraint, where we say a path is of $\gamma$ hops if it contains $\gamma$ edges.  We design a fully polynomial approximation scheme under hop constraints, and show that the algorithm can extend to run in pseudo-polynomial time under an additional linear constraint.  Further, we show how our algorithm can be applied to the cost-to-time ratio objective and the deadline problem.

\section{Related Work}
The route planning problem with nonlinear objective we consider here is related to multi-criteria optimization 
(e.g., \citeauthor{Papadimitriou:2000aa}~\citeyear{Papadimitriou:2000aa}).  The solutions to these problems typically look for approximate Pareto sets, namely a short list of feasible solutions that provide a good approximation to any non-dominated or optimal solution. 
A fully polynomial approximation scheme (FPTAS) for route planning with general monotonic and smoothed objective functions was presented by \citeauthor{Tsaggouris:2009aa}~\shortcite{Tsaggouris:2009aa}, based on  finding an approximate Pareto set for the feasible paths.  Later, \citeauthor{Mittal:2013aa}~\shortcite{Mittal:2013aa} provided a general approach for monotonic and smoothed combinatorial optimization problems and 
showed how to reduce these problems to multi-criteria optimization.

A special class of nonlinear optimization is \emph{quasi-concave minimization}.  This class of problems has the property that the optimal solution is an extreme point of the feasible set~\cite{Horst:2000aa}.  For quasi-concave combinatorial minimization problems, \citeauthor{Kelner:2007aa}~\shortcite{Kelner:2007aa} proved $(\log n)$-hardness of approximation and they also gave an additive approximation algorithm for low-rank quasi-concave functions based on smoothed analysis.  
Later, \citeauthor{Nikolova:2010aa}~\shortcite{Nikolova:2010aa} presented an FPTAS for risk-averse quasi-concave objective functions.  In recent work, \citeauthor{Goyal:2013aa}~\shortcite{Goyal:2013aa} gave an FPTAS for general monotone and smooth quasi-concave functions.

Without assuming monotonicity, the problem was studied by \citeauthor{Nikolova:2006ab}~\shortcite{Nikolova:2006ab}, who proved hardness results and gave a pseudo-polynomial approximation algorithm with integer edge weights for some specific objective functions.

\section{Problem Statement and Preliminaries}
\label{sec:problem_setting}
Consider a directed graph $G=(V,E)$.  Denote the number of vertices as $|V|=n$, and the number of edges as $|E|=m$.  Suppose we are given a source node $s$ and a destination node $t$.  In the \emph{nonlinear objective shortest path} problem, we seek an $s-t$ path $p$, possibly containing loops, that minimizes some nonlinear objective function $f(p)$.

More precisely, denote the feasible set of all $s-t$ paths by $\paths_t$, and let $\paths=\bigcup_{v \in V} \paths_v$ be the set of paths that start from $s$.  Here we have not yet imposed assumptions that $\paths$ contains only simple paths.  Let $f: \paths \rightarrow \R_+$ be a given objective function.  Then we can write the problem as finding the path $p^*$ that minimizes this objective function:
\begin{equation} \label{eq:prob_form}
p^* = \arg\min_{p \in \paths_t} f(p)
\end{equation}

We assume that $f$ is a low-rank nonlinear function with some smoothness condition.  Specifically, we say that $f$ is of rank $d$ if there exists a function $g: \R_+^d \rightarrow \R_+$ and an $m \times d$ weight matrix $\Wmat = (w_{e,k})_{m \times d}$, with positive entries, such that for any path $p \in \paths$, $f(p) = g(\sum_{e \in p} w_{e,1}, \sum_{e \in p} w_{e,2},  \dots, \sum_{e \in p} w_{e,d})$.

We interpret this as having $d$ criteria, such as for example length in miles, cost, travel time, etc. so that each edge is associated with $d$ additive 
weights, where the $k$-th weight for edge $e$ is denoted by $w_{e,k}$.  Then, the $k$-th weights of all edges can be represented by the column vector $\wvec_k = (w_{1,k}, w_{2,k}, \dots, w_{m,k})^T$.  For simplicity, we denote the sum of the $k$-th weights of the edges along path $p$ as $l_k(p) = \sum_{e \in p} w_{e,k}$, and we call it \emph{the $k$-th criterion of path $p$}.
Hence, the objective function $f$ can be rewritten as a function of all criteria: 
\begin{equation}\label{eq:deflowrank}
f(p) = g(l_1(p), l_2(p),  \dots, l_d(p) )
\end{equation}

Consider the linear space spanned by the $d$ criteria $\Wspace = span(\wvec_1, \wvec_2, \dots, \wvec_d)$.  Then, each feasible $s-t$ path $p \in \paths_t$ has an associated point in this $d$-dimensional space $\Wspace$ given by 
$\lvec(p)=(l_1(p), l_2(p), \dots, l_d(p))$, providing the values of the $d$ criteria for that path.  
We shall call the convex hull of 
such points the \emph{path polytope}.\footnote{Strictly speaking, it is the projection of the feasible path polytope in $\R^m$ onto this $d$-dimensional subspace but for simplicity we will just call it the path polytope in this paper.}  
We assume that the weights matrix $\Wmat$ contains only positive entries, and thus the path polytope is bounded from below: for any path $p \in \paths_t$ and criterion $k \in \{ 1,\dots,d\}$, we have $l_k(p) \ge \min_{p' \in \paths_t} l_k(p')$. This lower-bound corresponds to the shortest path with respect to the $k$-th criterion.  For upper bounds of the polytope, we defer the discussion to the later sections.  

\begin{figure}[t]
    \centering
    \includegraphics[width=\linewidth]{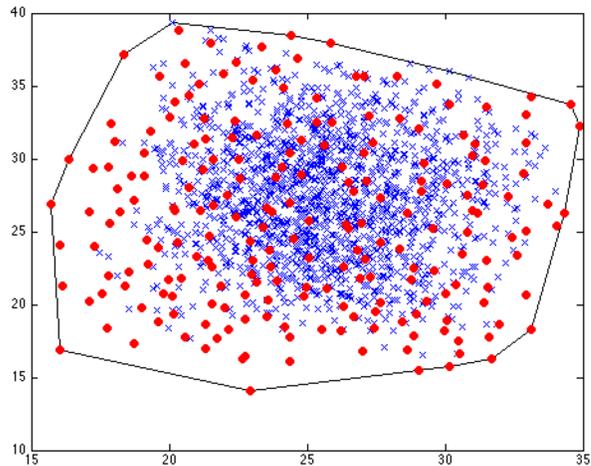}
    \caption{Illustration of a path polytope on a doubly weighted graph with a hop constraint.  Each blue cross and red dot correspond to an $s-t$ path projected onto $\Wspace$.  We can use polynomially many red dots to approximate the path polytope and consequently to approximate the optimal path.}
    \label{fig:approx_pathpolytope}
\end{figure}

The size of the path polytope 
is in worst case exponential in the problem size.  The main idea of this paper is to approximate the path polytope with polynomially many paths, as illustrated in Figure~\ref{fig:approx_pathpolytope}.  As we will show in this paper, the optimal path can be approximated if there are some smoothness conditions on the objective function.
The smoothness assumption we adapt here is as follows: $g$ is assumed to be $\beta$-Lipschitz on the log-log scale with respect to the $L1$-norm:  that is, for any $\yvec, \yvec'$ in the path polytope, 
\begin{equation} \label{eq:lipschitz}
| \log g(\yvec) - \log g(\yvec') | \le \beta \| \log \yvec - \log \yvec' \|_1
\end{equation}
where $\log \yvec$ is defined as a column vector: $\log \yvec = (\log y_1, \log y_2, \dots, \log y_d)^T$.

In this paper, we design approximation algorithms for this problem, formally defined as follows:
\begin{definition}
An $\alpha$-approximation algorithm for a minimization problem with optimal solution value $OPT$ and $\alpha > 1$ is a polynomial-time algorithm that returns a solution of value at most $\alpha OPT$.
\end{definition}
The best possible approximation algorithms provide solutions that are arbitrarily close to the optimal solution:
\begin{definition}
A fully-polynomial approximation scheme (FPTAS) is an algorithm for an optimization problem that, given desired accuracy $\epsilon>0$, finds in time polynomial in $\frac{1}{\epsilon}$ and the input size, a solution of value $OPT'$ that satisfies $|OPT-OPT'| \le \epsilon OPT$, for any input, where $OPT$ is the optimal solution value.
\end{definition}


\section{Hardness Results}
We can categorize the problem on general directed graphs into two types: one is the problem that accepts only simple paths, by which we mean each vertex can only be visited once in any feasible path; the other is the problem that accepts non-simple paths.  Unfortunately, both types of problems turn out to be hard.  In fact, the problem that only accepts simple paths in general involves finding the longest simple path, which is strongly NP-hard, namely it is NP-hard to approximate to within any constant factor \cite{Karger:1997aa}.  

First, we present a natural objective function that has been considered in the literature, that is hard to approximate.

\subsection{Minimum Cost-to-Time Ratio Path Problem}
The cost-to-time ratio objective function is a rank $2$ function, defined as follows:
\begin{equation*}
f(p)=g\bigg(\sum_{e\in p} w_{e,1}, \sum_{e\in p} w_{e,2}\bigg) = \frac{\sum_{e\in p} w_{e,1}}{\sum_{e\in p} w_{e,2}}
\end{equation*}
For simplicity, we assume $\Wmat$ has strictly positive entries.  This objective function has been considered previously for different combinatorial optimization problems \cite{Megiddo:1979aa,Correa:2010aa}.  However, finding the minimum cost-to-time ratio simple path is known to be an NP-hard problem \cite{Ahuja:1983aa}.  Here, we show its hardness of approximation:

\begin{theorem}
There is no polynomial-time $\alpha$-approximation algorithm for the minimum cost-to-time ratio simple path problem for any constant factor $\alpha$, unless $P=NP$.
\end{theorem}
\begin{proof}
To facilitate the presentation of our hardness of approximation proof, we first show that this problem is NP-hard.  We reduce a longest path problem instance to a minimum cost-to-time ratio simple path problem instance.  The longest path problem is known to be NP-hard to approximate to within any constant factor~\cite{Karger:1997aa}.  First, we choose some $\lambda>0$.  For an instance of the longest path problem, which is some graph $G$ with unweighted edges and a designated source node $s$ and sink node $t$, we create an instance of the minimum cost-to-time ratio path problem with the same graph $G$, source $s$ and sink $t$ and the following (cost,time) edge weights:
\begin{enumerate}
\item For each edge $e$ incident to $s$, set $w_{e}=(\lambda n+1,1)$.
\item For the remaining edges $e$, set $w_{e}=(1,1)$.
\end{enumerate}
We can see that a path is optimal under the ratio objective if and only if it is the longest path with respect to the unweighted edges (or, equivalently, with respect to the second weight).  This shows that the minimum cost-to-time ratio simple path problem is NP-hard.

To show the hardness of approximation, similarly, we use the fact that the longest path problem cannot be approximated to within any constant factor unless $P=NP$ \cite{Karger:1997aa}.  Suppose the contrary, namely that we can find a $(1+\epsilon)$-approximation algorithm for the minimum cost-to-time ratio simple path problem, for some $\epsilon \in (0, \lambda)$.  To be precise, since we can choose $\lambda$ while $\epsilon$ comes from the assumed approximation algorithm, we can set $\lambda > \epsilon$. 
For a given instance of longest path, consider the corresponding instance of minimum cost-to-time ratio path problem above. Let $p^*$ be the optimal min cost-to-time ratio path.  Then, the $(1+\epsilon)$-approximation algorithm will return a simple path $p$ such that
\begin{equation*}
\frac{\lambda n+l_2(p)}{l_2(p)} \le (1+\epsilon)\frac{\lambda n+l_2(p^*)}{l_2(p^*)}
\end{equation*}
Rearranging this, we get the inequality:
\begin{equation*}
(\lambda n-\epsilon l_2(p))l_2(p^*) \le \lambda(1+\epsilon)n l_2(p)
\end{equation*}
Since $l_2(p)\le n$ and $\epsilon<\lambda$, we get the following:
\begin{equation*}
l_2(p^*) \le \frac{\lambda(1+\epsilon)}{\lambda-\epsilon} l_2(p)
\end{equation*}
This shows that we can approximate the longest path problem to within a constant factor of $\frac{\lambda(1+\epsilon)}{\lambda-\epsilon}$, a contradiction. Therefore, the minimum cost-to-time ratio simple path problem cannot be approximated to within any constant factor unless $P=NP$.
\end{proof}

\subsection{General Objective Function}

For general objective functions, the following hardness results were proven by~\citeauthor{Nikolova:2006ab} \shortcite{Nikolova:2006ab}:

\begin{fact}
Suppose function $g$ attains a global minimum at $\yvec>0$.  Then problem \eqref{eq:prob_form} of finding an optimal $s-t$ simple path is NP-hard.
\end{fact}

\begin{fact}
Consider problem \eqref{eq:prob_form} of finding an optimal $s-t$ simple path with $f$ being a rank-$1$ function.  Let $p^*$ be the optimal path, and $l^*=l(p^*)$.  If $g$ is any strictly decreasing and positive function with slope of absolute value at least $\lambda>0$ on $[0, l^*]$, then there does not exist a polynomial-time constant factor approximation algorithm for finding an optimal simple path, unless $P=NP$.
\end{fact}

Although these facts are proven for rank-$1$ functions, they show that for a large class of objective functions, the problem of finding the optimal simple path is not only NP-hard, but also NP-hard to approximate within any constant factor.  

As a result, we turn our attention to the problem that accepts non-simple paths.  While this problem is also NP-hard~\cite{Nikolova:2006ab}, we are able to design a polynomial-time algorithm to approximate the optimal solution arbitrarily well.

\section{Approximation Algorithms}
Consider problem \eqref{eq:prob_form} that accepts non-simple paths.
For some cases, it is possible that there does not exist an optimal path of finite length.  Consider an example of a rank $1$ monotonic decreasing function $g(x)=1/x$.  Suppose there is a positive weight loop in any of the $s-t$ paths.  Then, we can always find a better solution by visiting such a loop arbitrarily many times: the optimal path here is ill-defined. 

Naturally, for the problem that accepts non-simple paths, the path polytope is infinite.  Even when the optimal path is finite, for a large class of problems, including the minimum cost-to-time ratio path problem, finding a finite subset of the path polytope that contains the optimal path is difficult, as shown in our Hardness section above.  In this section, we will show that, under some additional constraints, this problem becomes well-defined and well-approximable.  In what follows, we identify the additional constraints and present the algorithms to solve the corresponding problems.

\subsection{Hop Constraint} \label{subsec:hc}
Suppose the path contains at most $\gamma$ hops, then the set of feasible 
paths $\paths_t$ becomes finite and the problem is well-defined.  Denote $hop(p)$ to be the number of hops in  path $p$, or equivalently, the number of edges in $p$.  Note that an edge can be counted multiple times if path $p$ contains loops.

To approximate the optimal path, we approximate the path polytope with polynomially many paths, and then search for the best one among them.  To approximate the path polytope, we first identify a region of weights space $\Wspace$ that contains all paths starting from $s$, which are within $\gamma$ edges.  We can upper- and lower-bound this region as follows: For each criterion $k = 1, \dots, d$, define $c_k^{\min} = \min_{e \in E} w_{e,k}$ and $c_k^{\max} = \max_{e \in E} w_{e,k}$.
Then for any path $p$ that starts from the source $s$ and contains at most $\gamma$ edges, we have $l_k(p) \in [c_k^{\min}, \gamma c_k^{\max}]$, for each criterion $k \in \{1,\dots,d\}$.

Next, we partition the region $\prod_{k=1}^{d}[c_k^{\min}, \gamma c_k^{\max}]$ of the weights space $\Wspace$ into polynomially many hypercubes.  Specifically, for each criterion $k \in \{1,2,\dots, d\}$, let $\epsilon>0$ be a given desired approximation factor, and $C_k = c_k^{\max}/c_k^{\min}$.  Define $\Gamma_k = \{0, 1, \dots, \lfloor \log_{(1+\epsilon)} (\gamma C_k) \rfloor\}$ as the set of points that evenly partitions the segment $[c_k^{\min}, \gamma c_k^{\max}]$ on the log scale.  Then, we generate a $d$-dimensional lattice $\Gamma = \Gamma_1 \times \Gamma_2 \times \cdots \times \Gamma_d$ according to these sets.  After that, we define a hash function $\hvec=(h_1, h_2, \dots, h_d):\paths \rightarrow \Gamma$ that maps a path to a point on the lattice $\Gamma$.  Each $h_k: \paths \rightarrow \Gamma_k$ is defined as $h_k(p) = \lfloor \log_{(1+\epsilon)} \frac{l_k(p)}{c_k^{\min}} \rfloor$.

The main part of the algorithm iterates over $i \in \{1,2,\dots,\gamma\}$, sequentially constructing $s-v$ sub-paths containing at most $i$ edges, for each vertex $v \in V$.  For each iteration $i$ and each vertex $v \in V$, we maintain at most $|\Gamma|$ configurations as a $d$-dimensional table $\Pi_v^i$, indexed by $\yvec \in \Gamma$.  Each configuration $\Pi_v^i(\yvec)$ can be either an $s-v$ path $p$ of at most $i$ edges and $h(p)=\yvec$, given there is one, or $null$ otherwise.  The table can be implemented in various ways. \footnote{In practice, we can either allocate the memory of the entire table in the beginning, or use another data structure that only stores the used entries of the table.  According to our experimental results, only a small portion of the table $\Pi_v^i$ will be used.}  Here we assume that given index $\yvec \in \Gamma$, both checking whether the configuration $\Pi_v^i(\yvec)$ is $null$, retrieving the path in $\Pi_v^i(\yvec)$, and inserting/updating a path in $\Pi_v^i(\yvec)$ can be done in $T_\Pi=T_\Pi(n,|\Gamma|)$ time.

At iteration $i$, for each vertex $u$, we will construct a set of $s-u$ sub-paths containing $i$ edges based on $(i-1)$-edge sub-paths we previously constructed.  
We loop over each vertex $v \in V$ and each edge $e=(v,u)$ incident to $v$, and we append the edge to the end of each $(i-1)$-edge $s-v$ path in $\Pi_v^{i-1}$.  However, not all generated sub-paths will be stored.  Consider any $s-u$ path $p$ generated with this procedure, it will be stored in table $\Pi_u^i$ only if $\Pi_u^i(\hvec(p))$ is empty.  With this technique, we can ensure that the number of sub-paths we record in each iteration is at most polynomial, because it is bounded by the size of the lattice $\Gamma$ times the number of vertices.  This procedure is outlined in Algorithm~\ref{alg:nosp_digraph}.

\begin{algorithm}
\caption{Nonlinear Objective Shortest Path with Hop-constraint}
\label{alg:nosp_digraph}
\begin{algorithmic}[1]
	\FOR{each $\yvec \in \Gamma$, $v \in V$, and $i=1,2,\dots,\gamma$}
		\STATE Set $\Pi_v^i(\yvec)=null$
	\ENDFOR
	\STATE Set $\Pi_s^0=\{s\}$.
	\FOR{ $i=1$ to $\gamma$} \label{algn:iter_loop}
	\STATE Set $\Pi_v^i = \Pi_v^{i-1}$ for each vertex $v \in V$.
	\FOR{ each vertex $v \in V$}
		\FOR{ each path $p \in \Pi_v^{i-1}$ with $i-1$ edges}
			\FOR{ each edge $e = (v,u) \in E$}
				\STATE Set $p'$ be the path that appends edge $e$ to the tail of $p$.
				\STATE Compute hash $\yvec = \hvec(p')$.
				\IF{ $\Pi_u^i(\yvec) = null$} \label{algn:feasible_check}
					\STATE $\Pi_u^i(\yvec) = p'$. \label{algn:assign}
				\ENDIF
			\ENDFOR
		\ENDFOR
	\ENDFOR
	\ENDFOR
	\RETURN $\max_{p \in \Pi_t^{\gamma}} f(p)$
\end{algorithmic}
\end{algorithm}


To prove the correctness of this algorithm, we show the following lemmas: 
\begin{lemma}\label{lemma:single_hop}
Let $p_u$ and $p_u'$ be two $s-u$ paths such that for each criterion $k \in \{1,2,\dots,d\}$, 
\begin{equation} \label{eq:l1eq1}
| \log l_k(p_u) - \log l_k(p_u') | \le \log(1+\epsilon).
\end{equation}
Suppose there is a vertex $v$ and an edge $e$ such that $e=(u,v) \in E$.  Let $p_v$ and $p_v'$ be the $s-v$ paths that append $(u,v)$ to the tails of $p_u$ and $p_u'$, respectively.  Then, for each criterion $k \in \{1,2,\dots,d\}$, 
$$
| \log l_k(p_v) - \log l_k(p_v') | \le \log(1+\epsilon)
$$
\end{lemma}
\begin{proof}
By \eqref{eq:l1eq1}, it directly follows that
\begin{align*}
l_k(p_v) &= l_k(p_u)+l_k(e) \le (1+\epsilon) l_k(p_u') + l_k(e) \\
&\le (1+\epsilon) ( l_k(p_u') + l_k(e) ) = (1+\epsilon) l_k(p_v') \\
l_k(p_v') &= l_k(p_u')+l_k(e) \le (1+\epsilon) l_k(p_u) + l_k(e) \\
&\le (1+\epsilon) ( l_k(p_u) + l_k(e) ) = (1+\epsilon) l_k(p_v)
\end{align*}
Combining the above two inequalities yields the proof.
\end{proof}


\begin{lemma} \label{lemma:multi_hops_general1}
For any $v \in V$, and for any path $s-v$ that is of at most $i$ edges, there exists an $s-v$ path $p' \in \Pi_v^i$ such that for each criterion $k \in \{1,2,\dots, d\}$,
\begin{equation*}
| \log l_k(p) - \log l_k(p') | \le i \log(1+\epsilon)
\end{equation*}
\end{lemma}
\begin{proof}
We prove this lemma by induction on the number of edges in a path.  
The base case is a direct result from the definition of the hash function $\hvec$.
For the induction step, consider vertex $v \in V$, and we would like to prove the statement is correct for any $s-v$ sub-paths of at most $i^*$ edges, given the statement is true for any $i<i^*$ and any vertices.  Let $p_v$ be any $s-v$ path no more than $i^*$ edges.  If $p_v$ is of $i' < i^*$ edges, then by the induction hypothesis, there exists an $s-v$ path $p_v'$ in $\Pi_v^{i'}$ such that $|\log(l_k(p_v)/l_k(p_v'))| \le i' \log(1+\epsilon)$.  The lemma is proved for this case since $p_v'$ is in $\Pi_v^{i^*}$ as well.  Next consider $p_v$ to be an $s-v$ path with $i^*$ edges, and suppose $p_v$ is formed by appending the edge $(u, v)$ to an $s-u$ path $p_u$.  By the induction hypothesis, we can see that there is an $s-u$ path $p_u' \in \Pi_u^{i^*-1}$ such that for each criterion $k \in \{1,2,\dots, d\}$,
\begin{equation*}
| \log l_k(p_u) - \log l_k(p_u') | \le (i^*-1) \log(1+\epsilon)
\end{equation*}
Then, with Lemma~\ref{lemma:single_hop}, as we form $s-v$ sub-paths $p_v$ and $p_v'$ by appending $(u,v)$ to the end of $p_u$ and $p_u'$ respectively, we have the following relationship for each criterion $k \in \{1,2,\dots, d\}$:
\begin{equation} \label{eq:bounded_ratio}
| \log l_k(p_v) - \log l_k(p_v') | \le (i^*-1) \log(1+\epsilon)
\end{equation}
Consider any $s-v$ path $p_v''$ such that $\hvec(p_v'')=\hvec(p_v')$.  
We must have $| \log l_k(p_v') - \log l_k(p_v'') | \le \log(1+\epsilon)$ for each criterion $k \in \{1,\dots,d\}$, according to the definition of the hash function $\hvec$.
Combining this with \eqref{eq:bounded_ratio}, we can see that even for the case that path $p_v'$ is not stored in the table $\Pi_v^{i^*}$, the path $p_v'' = \Pi_v^{i^*}(h(p_v'))$ must satisfy
\begin{equation*}
|\log l_k(p_v) - \log l_k(p_v'') | \le i^* \log(1+\epsilon)
\end{equation*}
for each criterion $k \in \{1,\dots,d\}$.  This indicates that for any $s-v$ path $p_v$ within $i^*$ edges, there exists a path $p_v''$ in $\Pi_v^{i^*}$ such that their ratio is bounded by $(1+\epsilon)^{i^*}$. Q.E.D.
\end{proof}

Now, we are ready to state and prove the main result:
\begin{theorem}\label{thm:main}
Suppose $g$ defined in \eqref{eq:deflowrank} is $\beta$-Lipschitz on the log-log scale.  Then Algorithm~\ref{alg:nosp_digraph} is a $(1+\epsilon)^{\beta d \gamma}$-approximation algorithm for the $\gamma$-hop-constrained nonlinear objective shortest path problem.  Besides, Algorithm~\ref{alg:nosp_digraph} has time complexity $O(\gamma m(\frac{\log (\gamma C)}{\log (1+\epsilon)})^d T_\Pi)$, where $C=\max_{k\in\{1,2,\dots,d\}} C_k$.
\end{theorem}
\begin{proof}
By Lemma~\ref{lemma:multi_hops_general1}, we can see that for any $s-t$ path $p$ of at most $\gamma$ edges, there exists an $s-t$ path $p' \in \Pi_t^\gamma$ such that
\begin{equation*}
| \log l_k(p) - \log l_k(p') | \le \gamma \log (1+\epsilon), \quad \forall k=1,2,\dots ,d
\end{equation*}

It then follows from \eqref{eq:lipschitz} that our algorithm finds a path $p'$ such that $| \log f(p') - \log f(p^*) | \le \beta d \gamma \log (1+\epsilon)$.
Then, we have
\begin{equation*}
f(p') \le (1+\epsilon)^{\beta d \gamma}f(p^*)
\end{equation*}
This shows that $p'$ is an $(1+\epsilon)^{\beta d \gamma}$-approximation to $p^*$.

Next, we show the time complexity.  It is easy to see that at each iteration, the algorithm generates at most $m|\Gamma|$ sub-paths.  Each sub-path requires at most $O(T_\Pi)$ time to operate.  Since $|\Gamma|=(\frac{\log (\gamma C)}{\log(1+\epsilon)})^d$, we can conclude that Algorithm~\ref{alg:nosp_digraph} runs in time $O(\gamma m(\frac{\log (\gamma C)}{\log (1+\epsilon)})^d T_\Pi)$.
\end{proof}

Note that we can obtain a $(1+\delta)$-approximation algorithm by setting $\delta = (1+\epsilon)^{\beta d \gamma}-1$.  With this setting, then Algorithm~\ref{alg:nosp_digraph} has time complexity $O(\gamma m(\frac{\beta d \gamma \log(\gamma C)}{\delta})^d T_\Pi)$.  If we assume $\gamma=O(n)$, then Algorithm~\ref{alg:nosp_digraph} is indeed an FPTAS.

\subsection{Linear Constraint} \label{subsec:lc}
The hop constraint we considered can be relaxed:  we now consider the more general problem that any returned path $p$ must satisfy an additional linear constraint 
$\sum_{e \in p} w_{e,d+1} \le b,$ 
for some budget $b>0$ and some weight vector $\wvec_{d+1} \in \R_{++}^m$ with strictly positive entries.  We define the $(d+1)$-th criterion $l_{d+1}(p) = \sum_{e \in p} w_{e, d+1}$, which can be interpreted as having a budget $b$ for a path.  To solve this problem, we can find an upper bound of the number of edges for any feasible $s-t$ path.  A simple upper bound can be given by 
\begin{equation} \label{eq:hop_bound}
\gamma=\bigg\lceil\frac{b - \sum_{e\in p_{d+1}^*} w_{e,d+1}}{\min_{e \in E} w_{e,d+1}} + hop_{\min} \bigg\rceil
\end{equation}
where $p_{d+1}^* = \arg\min_{p\in\paths} \sum_{e\in p} w_{e,d+1}$ is the shortest path with respect to the weight $\wvec_{d+1}$, and $hop_{\min}$ is the minimum number of edges for any $s-t$ path.  Note that either $hop_{\min}$ and $p_{d+1}^*$ can be solved by Dijkstra's or Bellman-Ford's shortest path algorithms efficiently.

Then, we can adapt the techniques discussed for the hop-constrained problem but with some modifications.  Instead of assigning the least hop path to a configuration, we now assign the path with the least $(d+1)$-th criterion.  Formally, the modification is to change line~\ref{algn:feasible_check} and line~\ref{algn:assign} of Algorithm~\ref{alg:nosp_digraph} to :  Assign $p'$ to $\Pi_u^i(\yvec)$ if $p'$ satisfies both of the following:
\begin{itemize}
\item $l_{d+1}(p') < b$
\item $\Pi_u^i(\yvec)$ is $null$ or $l_{d+1}(p') < l_{d+1}(\Pi_u^i(\yvec))$
\end{itemize}

As a corollary of Theorem~\ref{thm:main}, we have the following result:
\begin{corollary} \label{cor:linear}
Algorithm~\ref{alg:nosp_digraph} with the modifications stated above is a $(1+\epsilon)^{\beta d \gamma}$-approximation algorithm for the linear-constrained nonlinear objective shortest path problem, where $\gamma$ is defined as in Eq. \eqref{eq:hop_bound}.
\end{corollary}

\section{Application: Deadline Problem}
With the techniques mentioned in the above section, we are able to solve various problems, including the minimum cost-to-time ratio path problem with either a hop or a linear budget constraint.  In this section, we discuss a related problem, which is the \emph{deadline problem} in a stochastic setting.

The deadline problem is to find a path that is the most likely to be on time, which is defined as follows: Suppose each edge is associated with an independent random variable $X_e$, which we call the travel time of edge $e$.  Given a deadline $D>0$, we are going to find an $s-t$ path such that the probability of being on time is maximized:  $\max_{p \in \paths_t} \Pr(\sum_{e\in p} X_e < D) $.
We assume that the random variables $\{X_e\}$ are Gaussian.  Based on this, we denote the mean travel time for edge $e$ as $\mu_e$ and variance as $\sigma_e^2$.  For simplicity, we assume for each edge $e \in E$, its mean $\mu_e$ and variance $\sigma_e^2$ are bounded by $\mu_e \in [\mu_{\min}, \mu_{\max}]$, and $\sigma_e^2 \in [\sigma_{\min}^2, \sigma_{\max}^2]$.   The problem can hence be written as:
\begin{equation}\label{eq:deadlineprob_def2}
\max_{p \in \paths_t} \Phi\bigg(\frac{D - \sum_{e\in p} \mu_e}{\sqrt{\sum_{e\in p} \sigma_e^2}}\bigg)
\end{equation}
where $\Phi(\cdot)$ is the cumulative distribution function for the standard normal distribution.  The objective function in \eqref{eq:deadlineprob_def2} is a rank $2$ function.  This problem is studied in \cite{Nikolova:2006aa} for a risk-averse shortest path problem, where they considered the case that the optimal path has mean travel time less than the deadline: for this case \eqref{eq:deadlineprob_def2} is a monotonic decreasing function.  In this paper, we consider the case that all feasible $s-t$ paths have mean travel time strictly greater than the deadline.  This means that we assume each path has poor performance.  Then, the objective function in \eqref{eq:deadlineprob_def2} becomes monotonic decreasing with respect to the mean but monotonic increasing with respect to the variance, for which known algorithms no longer apply.  However, we are able to use techniques presented in this paper to solve this problem under either a hop or an additional linear constraint.
\begin{theorem}\label{thm:deadline}
Suppose each $s-t$ path $p \in \paths_t$ satisfies $\sum_{e \in p} \mu_e > D$ and $\sum_{e \in p} \sigma_e^2 > S$.  For the deadline problem with a hop constraint $\gamma$, if Algorithm~\ref{alg:nosp_digraph} returns some path $p'$ such that $f(p') > \Phi(-3) \approx 0.0013$, then $p'$ is an $\alpha$-approximation for the deadline problem, where
$\alpha = \min\{384.62, (1+\epsilon)^{6.568(3+\frac{D}{\sqrt{S}})\gamma})\}$.  Similarly, if $p'$ is such that $f(p') > \Phi(-2) \approx 0.023$, then $\alpha = \min\{21.93, (1+\epsilon)^{4.745(2+\frac{D}{\sqrt{S}})\gamma}\}$.
\end{theorem}
Note that this theorem does not guarantee Algorithm~\ref{alg:nosp_digraph} admits FPTAS for the deadline problem, since given $\alpha$, the running time is polynomial in the factor $D/\sqrt{S}$.

\section{Experimental Evaluations}

\begin{figure}[t]
    \centering
    \includegraphics[width=\linewidth]{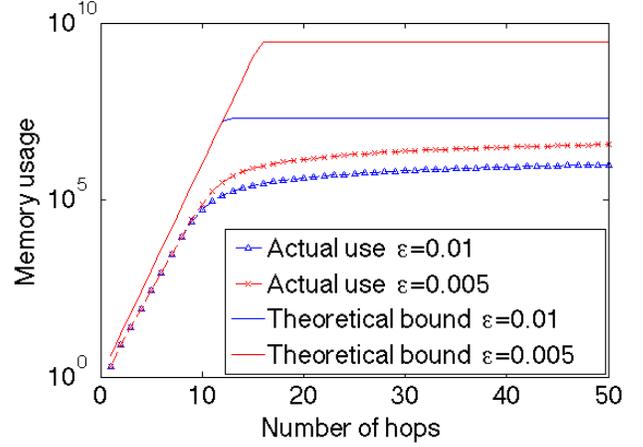}
    \caption{Memory usage for running our algorithm on a $5\times 5$ grid graph to the number of iterations.}
    \label{fig:memory}
\end{figure}

\begin{table}[t]
    \centering
    \caption{Deadline Problem with Hop Constraint}
    \label{table:deadline}
    \footnotesize
    \begin{tabular}{l | c | r | r | r}
        \hline
        Network	& Hops & Worst-case  		& Memory	& Run \\
        			&	& Theoretical		& Usage & time (s)\\
        			& 	& Accuracy ($\alpha$)& ($10^6$ paths) &  \\
        \hline \hline
        $5 \times 5$ Grid	& $12$ & (Exhaustive) & $1.63$ & $3.441$ \\
        \cline{2-5}
        					& $12$ & $2.2923$ & $0.67$ & $2.057$ \\
        \cline{2-5}
        					& $15$ & $2.2926$ & $4.30$ & $20.276$ \\
        \hline
        $10 \times 10$ Grid 	& $15$ & $2.2891$ & $18.05$ & $83.430$ \\
        \hline
        Anaheim 			& $15$ & $3.0436$ & $9.60$ & $37.671$\\
        \hline
    \end{tabular}
\end{table}

We tested our algorithm on a grid graph and a real transportation network.  The algorithm is implemented in C++.  Experiments are run on an Intel Core i7 1.7 GHz (I7-4650U) processor.  Due to memory constraints, the table $\{\Pi_v^i|v \in V\}$ is implemented as a binary search tree, where query, insertion, and update operations can be done in time $T_\Pi = O(\log(n|\Gamma|))$.  Memory turns out to be the limiting factor for the problem size and accuracy available to our experiments.  In this section, \emph{memory usage} refers to the total number of sub-paths maintained in the table $\{\Pi_v^i|v \in V\}$.

We tested the deadline problem with a hop constraint on a $5 \times 5$ grid graph, $10 \times 10$ grid graph, and the Anaheim network ($416$ vertices, $914$ edges) from the real world data set \cite{Bar-Gera:2002aa}.  The grid graphs are bi-directional, and the mean and the variance for each edge are randomly generated from $[0.1, 5]$.  We routed from node $(0,0)$ to node $(4,4)$ on the $5 \times 5$ grid, and from $(1,1)$ to $(8,8)$ on the $10 \times 10$ grid.  The Anaheim dataset provides the mean travel time. The variance for each edge is randomly generated from $0$ to the mean.  The source and destination are randomly chosen.

Table~\ref{table:deadline} summarizes the results on the datasets, which are the average of $20$ experiments.  Even though theoretically the algorithm can be designed for arbitrary desired accuracy, it is limited by the memory of the environment.\footnote{It takes about $3$GB physical memory usage for the $18$M paths in the $10\times10$ grid network.}  However, comparing with exhaustive search, we found that the practical accuracy is $100\%$ on $5 \times 5$ grid graphs.\footnote{Exhaustive search becomes prohibitive for paths with more than $12$ hops on $5 \times 5$ grid graphs.}
This means that our algorithm is far more accurate than its theoretical worst case bound.  We explain this phenomenon through Figure~\ref{fig:memory}, which shows the number of sub-paths generated up to each iteration.  We see that the number of sub-paths generated in the initial iterations grows exponentially, just as in exhaustive search, but levels off eventually.  Recall Lemma~\ref{lemma:multi_hops_general1}, which tells us that an error is induced when a sub-path is eliminated due to similarity to a saved sub-path, and Figure~\ref{fig:memory} suggests that this rarely happens in the initial iterations.


\section{Conclusion}
We investigated the shortest path problem with a general nonlinear and non-monotonic objective function.  We proved that this problem is hard to approximate for simple paths and consequently, 
we designed and analyzed a fully polynomial approximation scheme for finding the approximately optimal non-simple path under either a hop constraint or an additional linear constraint.  We showed that this algorithm can be applied to the cost-to-time ratio problem and the deadline problem.  Our experiments show that our algorithm is capable of finding good approximations efficiently.  

\section{Acknowledgement}
This work was supported in part by NSF grant numbers CCF-1216103, CCF-1350823, CCF-1331863, and a Google Faculty Research Award.

\bibliography{nonlinear_sp_arxiv} 
\bibliographystyle{aaai}

\section{Appendix: Missing Proofs}
\subsection{Proof of Corollary~\ref{cor:linear}}
We need the following modified version of Lemma~\ref{lemma:multi_hops_general1} to prove Corollary~\ref{cor:linear}:
\begin{lemma} \label{lemma:multi_hops_linear}
For any $v \in V$, and for any path $s-v$ that is of at most $i$ edges, there exists an $s-v$ path $p' \in \Pi_v^i$ such that
\begin{equation*}
\frac{1}{(1+\epsilon)^i} \le \frac{l_k(p)}{l_k(p')} \le (1+\epsilon)^i, \quad \forall k=1,2,\dots ,d
\end{equation*}
and
\begin{equation*}
l_{d+1}(p') \le l_{d+1}(p)
\end{equation*}
\end{lemma}
\begin{proof}
We prove this lemma by induction on the number of edges in a path as what we have done in Lemma~\ref{lemma:multi_hops_general1}.  The base case is a direct result from the definition of the hash function $\hvec$.  For the induction step, consider vertex $v \in V$, and we would like to prove the statement is correct for any $s-v$ sub-paths of at most $i^*$ edges, given the statement is true for any $i<i^*$ and any vertices.  

Let $p_v$ be any $s-v$ path of no more than $i^*$ edges.  First consider the case that the path $p_v$ is of $i' < i^*$ edges, then by induction hypothesis, there exists an $s-v$ path $p_v'$ in $\Pi_v^{i'}$ such that $|l_k(p_v)/l_k(p_v')| \le (1+\epsilon)^{i'}$.  Then consider the $s-v$ path $p_v'' \in \Pi_v^{i^*}$ such that $\hvec(p_v'')=\hvec(p_v')$.  By the definition of the hash function $\hvec$, it must hold that
\begin{equation*}
\frac{1}{1+\epsilon} \le \frac{l_k(p_v')}{l_k(p_v'')} \le 1+\epsilon, \quad \forall k=1,2,\dots ,d
\end{equation*}
Then we have
\begin{equation*}
\frac{1}{(1+\epsilon)^{i'+1}} \le \frac{l_k(p_v)}{l_k(p_v'')} \le (1+\epsilon)^{i'+1}, \quad \forall k=1,2,\dots ,d
\end{equation*}
By the induction hypothesis and the replacement rule in the algorithm, we have $l_{d+1}(p_v'') \le l_{d+1}(p_v') \le l_{d+1}(p_v)$.  Therefore the lemma is proved for this case.

Next consider the case that $p_v$ is $s-v$ path with $i^*$ edges, and suppose $p_v$ is formed by appending the edge $(u, v)$ to an $s-u$ path $p_u$.  By the induction hypothesis, we can see that there is an $s-u$ path $p_u' \in \Pi_u^{i^*-1}$ such that
\begin{equation*}
\frac{1}{(1+\epsilon)^{i^*-1}} \le \frac{l_k(p_u)}{l_k(p_u')} \le (1+\epsilon)^{i^*-1}, \quad \forall k=1,2,\dots ,d
\end{equation*}
and $l_{d+1}(p_u') \le l_{d+1}(p_u)$.
Then, with Lemma~\ref{lemma:single_hop}, as we form the $s-v$ sub-paths $p_v$ and $p_v'$ by appending $(u,v)$ to the end of $p_u$ and $p_u'$ respectively, we have the following relationship:
\begin{equation} \label{eq:bounded_ratio2}
\frac{1}{(1+\epsilon)^{i^*-1}} \le \frac{l_k(p_v)}{l_k(p_v')} \le (1+\epsilon)^{i^*-1}, \quad \forall k=1,2,\dots ,d
\end{equation}
Consider any $s-v$ path $p_v''$ such that $\hvec(p_v'')=\hvec(p_v')$.  The following is always satisfied according to the definition of the hash function $\hvec$:
\begin{equation*}
\frac{1}{1+\epsilon} \le \frac{l_k(p_v')}{l_k(p_v'')} \le 1+\epsilon, \quad \forall k=1,2,\dots ,d
\end{equation*}
Combining this with \eqref{eq:bounded_ratio2}, we can see that the path $p_v'' = \Pi_v^{i^*}(h(p_v'))$ must satisfy
\begin{equation*}
\frac{1}{(1+\epsilon)^{i^*}} \le \frac{l_k(p_v)}{l_k(p_v'')} \le (1+\epsilon)^{i^*}, \quad \forall k=1,2,\dots ,d
\end{equation*}
and $l_{d+1}(p_v'') \le l_{d+1}(p_v)$.
This indicates that for any $s-v$ path $p_v$ which is within $i^*$ edges, there exists a path $p_v''$ in $\Pi_v^{i^*}$ such that their ratio is bounded by $(1+\epsilon)^{i^*}$, and $p_v''$ always has the smaller $(d+1)$-th weight than $p_v$. Q.E.D.
\end{proof}

With this lemma, we are able to prove Corollary~\ref{cor:linear}:
\begin{proof}[Proof of Corollary~\ref{cor:linear}]
By Lemma~\ref{lemma:multi_hops_linear}, we can see that for any $s-t$ path $p$ that is within $\gamma$ edges, there exists an $s-t$ path $p' \in \Pi_t^\gamma$ such that
\begin{equation*}
\bigg| \log \frac{l_k(p)}{l_k(p')} \bigg| \le \gamma \log (1+\epsilon) \quad \forall k=1,2,\dots ,d
\end{equation*}
and $l_{d+1}(p') \le l_{d+1}(p)$.  It then follows from \eqref{eq:lipschitz} that our algorithm finds a path $p'$ such that
\begin{equation*}
\bigg| \log \frac{f(p')}{f(p^*)} \bigg| \le \beta d \gamma \log (1+\epsilon)
\end{equation*}
Then, we have
\begin{equation*}
f(p') \le (1+\epsilon)^{\beta d \gamma}f(p^*)
\end{equation*}
Clearly $p'$ satisfies the constraint $l_{d+1}(p') \le b$ since $l_{d+1}(p') \le l_{d+1}(p^*) \le b$.  This shows that $p'$ is an $(1+\epsilon)^{\beta d \gamma}$-approximation to $p^*$.
\end{proof}

\subsection{Proof of Theorem~\ref{thm:deadline}}
Before we prove Theorem~\ref{thm:deadline}, we need the following lemma:
\begin{lemma} \label{lemma:norm_lipschitz}
The function $\Phi(\frac{D-x}{\sqrt{y}})$ is $3.284(3+\frac{D}{\sqrt{S}})$-Lipschitz on log-log scale with respect to $L1$ norm on the region $\Cregion = \{(x,y)|x>D, y>S, x-D\le 3\sqrt{y}\}$.
\end{lemma}
\begin{proof}
We want to show that for any $(u_1, v_1), (u_2, v_2) \in \Cregion' = \{(u,v)| e^u>D, e^v>S, e^u-D \le 3\sqrt{e^v}\}$, 
\begin{align}
&\bigg|\log \Phi\bigg(\frac{D-e^{u_1}}{\sqrt{e^{v_1}}}\bigg) - \log \Phi\bigg(\frac{D-e^{u_2}}{\sqrt{e^{v_2}}}\bigg)\bigg| \nonumber \\
& \quad \le 3.284\bigg(3+\frac{D}{\sqrt{S}}\bigg) (|u_1-u_2| + |v_1-v_2|) \label{eq:norm_lipschitz}
\end{align}
We take the derivatives from the function $\log \Phi(\frac{D-e^{u}}{\sqrt{e^{v}}})$:
\begin{align}
\bigg| \frac{\partial}{\partial u} \log \Phi\bigg(\frac{D-e^{u}}{\sqrt{e^{v}}}\bigg) \bigg| &= \bigg| \frac{e^u}{\sqrt{e^v}} \cdot \frac{\Phi'\bigg(\frac{D-e^{u}}{\sqrt{e^{v}}}\bigg)}{\Phi\bigg(\frac{D-e^{u}}{\sqrt{e^{v}}}\bigg)} \bigg| \nonumber \\
& \le 3.284 \bigg( \frac{e^u - D}{\sqrt{e^v}} + \frac{D}{\sqrt{e^v}} \bigg) \label{eq:normfunc_ratio} \\
& \le 3.284 \bigg( 3 + \frac{D}{\sqrt{e^v}} \bigg) \label{eq:normfunc_2} \\
& \le 3.284 \bigg( 3 + \frac{D}{\sqrt{S}} \bigg) \label{eq:normfunc_3} \\
\bigg| \frac{\partial}{\partial v} \log \Phi\bigg(\frac{D-e^{u}}{\sqrt{e^{v}}}\bigg) \bigg| &= \bigg| e^v \cdot \frac{D-e^u}{2\sqrt{e^v}^3} \cdot \frac{\Phi'\bigg(\frac{D-e^{u}}{\sqrt{e^{v}}}\bigg)}{\Phi\bigg(\frac{D-e^{u}}{\sqrt{e^{v}}}\bigg)} \bigg| \nonumber \\
& \le \frac{3}{2} \times 3.284 \label{eq:normfunc_4}
\end{align}
where Inequality \eqref{eq:normfunc_ratio} comes from the fact that $\Phi'(x)/\Phi(x) \le 3.284$ for $x \in [-3, 0)$, Inequality \eqref{eq:normfunc_2} comes from the fact that $e^u-D<3\sqrt{e^v}$.  Since $\log \Phi(\frac{D-e^{u}}{\sqrt{e^{v}}})$ is differentiable in $\Cregion'$, we obtain \eqref{eq:norm_lipschitz}.
\end{proof}
\begin{proof}[Proof of Theorem~\ref{thm:deadline}]
Let $l_1(p)$ denote the mean traveling time for path $p$, and $l_2(p)$ denote the variance for path $p$.  Further denote $\lvec(p)=(l_1(p), l_2(p))$.
By Lemma~\ref{lemma:multi_hops_general1}, we can see that for the optimal $s-t$ path $p^*$, there exists an $s-t$ path $p'' \in \Pi_t^\gamma$ such that
\begin{equation*}
\bigg| \log \frac{l_k(p^*)}{l_k(p'')} \bigg| \le \gamma \log (1+\epsilon)
\end{equation*}
Denote the region $\Cregion = \{(x,y) | x>D, y>S, x-D \le 3\sqrt{y}\}$.  We can see that $(x,y) \in \Cregion$ if and only if $\Phi(\frac{D-x}{\sqrt{y}}) \in [\Phi(-3),0.5)$ and $y>S$.  The optimal path must satisfy $\lvec(p^*) \in \Cregion$ since $f(p^*) \ge f(p')>\Phi(-3)$.

First consider the case that $\lvec(p'')$ is in the region $\Cregion$.  Then we can adapt the same technique as in the proof of Theorem~\ref{thm:main} and the smoothness condition we showed in Lemma~\ref{lemma:norm_lipschitz} to claim that the path $p'$ returned by the algorithm satisfies
\begin{align*}
f(p^*) &\le (1+\epsilon)^{2\times 3.284\times (3+\frac{D}{\sqrt{S}}) \gamma} f(p'') \\
&\le (1+\epsilon)^{2\times 3.284\times (3+\frac{D}{\sqrt{S}}) \gamma} f(p')
\end{align*}

Now consider the case that $\lvec(p'')$ is not in the region $\Cregion$, then there must exist a point $(x_0, y_0) \in \Cregion$ on the line $L=\{x-D = 3\sqrt{y}\}$ such that
\begin{align*}
\bigg| \log \frac{x_0}{l_1(p^*)} \bigg| \le \gamma \log (1+\epsilon) \\
\bigg| \log \frac{y_0}{l_2(p^*)} \bigg| \le \gamma \log (1+\epsilon)
\end{align*}
It then follows from Lemma~\ref{lemma:norm_lipschitz} that 
\begin{equation*}
\bigg| \log \frac{\Phi(\frac{D-x_0}{\sqrt{y_0}})}{f(p^*)} \bigg| \le 2\times 3.284 \times \bigg(3+\frac{D}{\sqrt{S}}\bigg) \gamma \log (1+\epsilon)
\end{equation*}
Since the algorithm returns some other path $p'$ such that $f(p')>\Phi(-3)=\Phi(\frac{D-x_0}{\sqrt{y_0}})$, we must have
\begin{align*}
f(p^*) &\le (1+\epsilon)^{6.568(3+\frac{D}{\sqrt{S}})\gamma} \Phi(-3) \\
&< (1+\epsilon)^{6.568(3+\frac{D}{\sqrt{S}})\gamma} f(p')
\end{align*}

However, the approximation ratio is upper-bounded by $0.5/\Phi(-3)\approx 384.62$, so this shows that $p'$ is an $\min\{384.62, (1+\epsilon)^{6.568(3+\frac{D}{\sqrt{S}})\gamma}\}$-approximation to $p^*$.  We can apply the same technique for the case that $f(p') > \Phi(-2)$ or any other values.
\end{proof}

\end{document}